\pdfoutput=1

\documentclass[11pt,letterpaper]{article}

\usepackage{mystyle} 
\title{On the Geometry of Nash and Correlated Equilibria with Cumulative Prospect Theoretic Preferences}
\author{Soham R.\ Phade and Venkat Anantharam
\thanks{Research supported by the NSF Science and Technology Center grant CCF-
0939370: "Science of Information", the NSF grants ECCS-1343398, CNS-1527846
and CIF-1618145, and the William and Flora Hewlett Foundation supported Center for Long Term Cybersecurity at Berkeley.}
\thanks{The authors are with the Department of Elelctrical Engineering and Computer Science, University of California, Berkeley, Berkeley, CA 94720.
        {\tt\small soham\_phade@berkeley.edu, ananth@eecs.berkeley.edu}}%
}
\date{}
\begin{document}
\maketitle

\begin{abstract}
It is known that the set of all correlated equilibria of an $n$-player non-cooperative game is a convex polytope and includes all the Nash equilibria. Further, the Nash equilibria all lie on the boundary of this polytope. We study the geometry of both these equilibrium notions when the players have cumulative prospect theoretic (CPT) preferences. The set of CPT correlated equilibria includes all the CPT Nash equilibria but it need not be a convex polytope. We show that it can, in fact, be disconnected. However, all the CPT Nash equilibria continue to lie on its boundary. We also characterize the sets of CPT correlated equilibria and CPT Nash equilibria for all $2\times 2$ games.
\end{abstract}


\section{Introduction}

Non-cooperative game theory studies the interaction between decision makers with possibly different objectives. In most cases, it is assumed that the decision makers are \emph{rational}. 
Rationality is generally formulated as expected utility maximization. The justification for this comes from the von Neumann and Morgenstern expected utility maximization theorem \citep{von2007theory}. Although this assumption has a nice normative appeal to it and can be used to a large extent as a prescriptive theory, it has been evident through several examples  \citep{allais1953comportement,fishburn1979two,kahneman1979prospect} that the model is not that good an approximation for descriptive purposes. On the other hand, cumulative prospect theory (CPT) accommodates many empirically observed behavioral features without losing much tractability \citep{tversky1992advances}. 

Two of the most well known notions of equilibrium, Nash equilibrium \citep{nash1951non} and correlated equilibrium \citep{aumann1974subjectivity}, are based on expected utility theory (EUT). \citet{keskin2016equilibrium} defines analogs for both these equilibrium notions based on cumulative prospect theory. He also establishes the existence of such equilibria under certain continuity conditions.

There has been considerable interest in the study of the comparative geometry of Nash and correlated equilibria. Under EUT, it is known that the set of all correlated equilibria is a convex polytope and contains the set of all Nash equilibria. In the paper by \citet{nau2004geometry}, it has been proved that: 
\begin{equation*}
\text{\parbox{.85\linewidth}{the Nash equilibria all lie on the boundary of the correlated equilibrium polytope.}}\tag{P}
\end{equation*}
Further, it has been found that in $2$-player (bimatrix) games, all extremal Nash equilibria are also extremal correlated equilibria \citep{cripps1995extreme,canovas1999nash,evangelista1996note}, although this result does not hold for more than $2$ players \citep{nau2004geometry}. The sets of correlated  and Nash equilibria have been completely characterized for $2 \times 2$ games \citep{calvo2006set}. 

Cumulative prospect theory is known to share common features with expected utility theory. The purpose of this paper is to study how the geometry of equilibrium notions is affected by prospect theoretic preferences. For example, under CPT, it continues to be the case that the set of correlated equilibria contains all Nash equilibria, but the set of correlated equilibria is not guaranteed to be a convex polytope  \citep[Example~2]{keskin2016equilibrium}. The pure Nash equilibria, if they exist, coincide under EUT and CPT \citep[Proposition~2]{keskin2016equilibrium}. It is known that the set of correlated equilibria under CPT includes the set of joint probability distributions induced by the convex hull of the set of pure Nash equilibria \citep[Proposition~3]{keskin2016equilibrium}, as is true under EUT. 

These similarities and differences raise the natural question of whether
property (P) continues to hold or not under CPT.
In fact, we will see that the set of correlated equilibria can be disconnected (section~\ref{sec: conn}). Nevertheless, 
our main result says that property (P) continues to hold under CPT (section~\ref{sec: main}). We also show that for $2 \times 2$ games the set of correlated equilibria under CPT is a convex polytope, and we characterize it (section~\ref{sec: 2x2games}).


\section{Preliminaries and main result}\label{sec: main}

Let $\Gamma = (N,(S_i)_{i \in N},(h_i)_{i \in N})$ be a finite $n$-person normal form game, where $N = \{1,\dots,n\}$ is the set of \emph{players}, $S_i$ is the finite strategy set of player $i \in N$, and $h_i: S_1 \times \dots \times S_n \to \bbR$ is the payoff function for player $i \in N$. Let each player $i \in N$ have at least two strategies, i.e $|S_i| \geq 2,\forall i \in N$. Let the set of all joint strategies be denoted by $S = \prod_{i \in N} S_i$. Let $s_i \in S_i$ denote a pure strategy of player $i \in N$ and let $s = (s_1,\dots,s_n) \in S$ denote a joint strategy of all players. Let $S_{-i} = \prod_{j \in N\back i} S_j$ denote the set of joint strategies $s_{-i} \in S_{-i}$ of all players except player $i$. Let $h_i(s)$ denote the payoff of player $i$ when joint strategy $s$ is played, and let $h_i(d_i,s_{-i})$ denote the payoff to player $i$ when she chooses strategy $d_i \in S_i$ while the others adhere to $s$.

\begin{definition}
	The game $\Gamma$ is \emph{non-trivial} if $h_i(s) \neq h_i(d_i,s_{-i})$ for some player $i \in N$, some $s \in S$, and some $d_i \in S_i$.
\end{definition}

Consider a joint probability distribution $\mu$ on $S$, viewed as a vector in $\bbR^{|S|}$ with coordinates $\mu(s) \geq 0$ for all $s \in S$, and satisfying $\sum_{s \in S} \mu(s) = 1$. Let $\Delta^{k}$ denote the standard $k$-simplex,
\[
	\Delta^k = \{(p_1,\dots,p_{k+1}) \in \bbR^{k+1} | p_1 + \dots + p_{k+1} = 1, p_i \geq 0,\forall 1 \leq i \leq k+1\},
\]
with the usual topology.

\begin{definition}[\cite{aumann1987correlated}]
	A joint probability distribution $\mu \in \Delta^{|S| - 1}$ is said to be a \emph{correlated equilibrium} of $\Gamma$ if it satisfies the following inequalities:
\begin{equation}
	\sum_{s_{-i} \in S_{-i}} \mu(s)(h_i(s_i,s_{-i}) - h_i(d_i,s_{-i})) \geq 0, \text{ for all $i$ and for all } s_i,d_i \in S_i.
\end{equation}
\end{definition}
The set of all correlated equilibria, henceforth denoted as $C_{EUT}$, is a convex polytope which is a proper subset of $\Delta^{|S|-1}$ iff the game is non-trivial. The set $I$ of all joint probability distributions that are of product form is defined by a system of nonlinear constraints, viz.
\[
	 I = \{\mu \in \Delta^{|S|-1}: \mu(s) = \mu_1(s_1) \times \dots \times \mu_n(s_n) \quad \forall \; s \in S\},
\]
where $\mu_i$ denotes the marginal probability distribution on $S_i$ induced by $\mu$. The set of all Nash equilibria is the intersection of $I$ and $C_{EUT}$, which is non-empty by virtue of Nash's existence theorem.

We now give a quick review of cumulative prospect theory (CPT) \citep[for more details see][]{wakker2010prospect}. Each person is associated with a \emph{reference point} $r \in \bbR$, a corresponding \emph{value function} $v^r : \bbR \to \bbR$, and two \emph{probability weighting functions} 
$w^\pm:[0,1] \to [0,1]$, $w^+$ for gains and $w^-$ for losses. The function $v^r(x)$ satisfies:
\begin{inparaenum}[(i)]
	\item it is continuous in $x$;
	\item $v^r(r) = 0$;
    \item it is strictly increasing in $x$.
\end{inparaenum}
The value function is generally assumed to be convex in the losses frame ($x < r$) and concave in the gains frame ($x > r$), and is steeper in the losses frame than in the gains frame in the sense that $v^r(r) - v^r(r -z) \ge v^r(r+z) - v^r(r)$ for all $z \ge 0$. The probability weighting functions $w^\pm: [0,1] \to [0,1]$ satisfy:
\begin{inparaenum}[(i)]
	\item they are continuous;
	\item they are strictly increasing;
	\item $w^\pm(0) = 0$ and $w^\pm(1) = 1$.
\end{inparaenum}

Suppose a person faces a \emph{lottery} (or \emph{prospect})
\[
	L = \{(p_1,z_1);\dots;(p_t,z_t)\},
\]
where $z_j \in \bbR, 1 \leq j \leq t$, denotes an \emph{outcome} and $p_j, 1 \leq j \leq t$, is the probability with which outcome $z_j$ occurs. We assume the lottery to be exhaustive, 
i.e. $\sum_{j=1}^t p_j = 1$ (Note that we are allowed to have $p_j =0$ for some values of $j$.) Let $z = (z_1,\dots,z_t)$ and $p = (p_1\dots,p_t)$. We denote $L$ as $(p,z)$ and refer to the vector $z$ as an \emph{outcome profile}.

Let $a = (a_1,\dots,a_t)$ be a permutation of $(1,\dots,t)$ such that
\begin{equation}\label{eq: order}
	z_{a_1} \geq z_{a_2} \geq \dots \geq z_{a_t}.
\end{equation}
Let $0 \leq j_r \leq t$ be such that $z_{a_j} \geq r$ for $1 \leq j \leq j_r$ and $z_{a_j} < r$ for $j_r < j \leq t$. (Here $j_r = 0$ when $z_{a_j} < r$ for all $1 \leq j \leq t$.) The \emph{CPT value} $V^r(L)$ of the prospect $L$ is evaluated using the value function $v^r(\cdot)$ and the probability weighting functions $w^{\pm}(\cdot)$ as follows:
\begin{equation}\label{eq: CPT_value_discrete}
	V^r(L) := \sum_{j=1}^{j_{r}} \pi_j^+(p,a) v^r(z_{a_j}) + \sum_{j=j_r+1}^t \pi_j^-(p,a) v^r(z_{a_j}),
\end{equation}
where $\pi^+_j(p,a),1 \leq j \leq j_{r}, \pi^-_j(p,a), j_r < j \leq t$, are \emph{decision weights} defined via:
\begin{align*}
	\pi^+_{1}(p,a) &= w^+(p_{a_1}),\\ 
	\pi_j^+(p,a) &= w^+(p_{a_1} + \dots + p_{a_{j}}) - w^+(p_{a_1} + \dots + p_{a_{j-1}}) &\text{ for } &1 < j \leq t, \\
	 \pi_j^-(p,a) &= w^-(p_{a_t} + \dots + p_{a_j}) - w^-(p_{a_t} + \dots + p_{a_{j+1}}) &\text{ for } &1 \leq j < t,\\
	 \pi^-_{t}(p,a) &= w^-(p_{a_t}). 
\end{align*}
Although the expression on the right in equation~(\ref{eq: CPT_value_discrete}) depends on the permutation $a$, one can check that the formula evaluates to the same value $V^r(L)$ as long as the permutation $a$ satisfies (\ref{eq: order}). The CPT value in equation~(\ref{eq: CPT_value_discrete}) can equivalently be written as:
\begin{align}\label{eq: CPT_value_cumulative}
	V^r(L) &= \sum_{j = 1}^{j_r - 1} w^+\l(\sum_{i = 1}^j p_{a_j}\r)\l[v^r(z_{a_j}) - v^r(z_{a_{j+1}})\r] \nonumber\\
	&+ w^+\l(\sum_{i = 1}^{j_r} p_{a_j}\r)v^r\l(z_{a_{j_r}}\r) + w^-\l(\sum_{i = j_r + 1}^{t} p_{a_j}\r)v^r(z_{a_{j_r+1}}) \nonumber \\
	&+ \sum_{j = j_r + 1}^{t-1} w^-\l(\sum_{i = j+1}^t p_{a_j}\r)\l[v^r(z_{a_{j+1}}) - v^r(z_{a_{j}})\r].
\end{align}

A person is said to have CPT preferences if, given a choice between prospect $L_1$ and prospect $L_2$, she chooses the one with higher $CPT$ value. 

CPT satisfies \emph{strict stochastic dominance} \citep{chateauneuf1999axiomatization}: shifting positive probability mass from an outcome to a strictly preferred outcome leads to a strictly preferred prospect. For example, the prospect $L_1 = \{(0.6,40);(0.4,20)\}$ can be obtained from the prospect $L_2 = \{(0.5,40);(0.5,20)\}$ by shifting a probability mass of $0.1$ from outcome $20$ to a strictly better outcome $40$. The strict stochastic dominance condition says that $V^r(L_1) > V^r(L_2)$ (see equation~(\ref{eq: CPT_value_cumulative})).

Also, CPT satisfies \emph{strict monotonicity} \citep{chateauneuf1999axiomatization}: any prospect becomes strictly better as soon as one of its outcomes is strictly improved. For example, if $L_1 = \{(0.6,40);(0.4,-10)\}$ and $L_2 = \{(0.6,40);(0.4,-20)\}$, then $V^r(L_1) > V^r(L_2)$ (see equation~(\ref{eq: CPT_value_discrete})).

We now describe the notion of correlated equilibrium incorporating CPT preferences, as defined by \citet{keskin2016equilibrium}. Let $\{v_i^r(\cdot),r \in \bbR\}$ be a family of value functions, one for each reference point, and $w_i^\pm(\cdot)$ be the probability weighting functions for each player $i \in N$. 
We assume that $v_i^r(x)$ is continuous in $x$ and $r$ for each $i$.
For every player $i \in N$, let the reference point be determined by a continuous function $r_i:\Delta^{|S|-1} \to \bbR$. 
Let $V_i^r(L)$ denote the CPT value of a lottery $L$ evaluated by player $i$, using the value function $v_i^r(\cdot)$ and probability weighting functions $w_i^\pm(\cdot)$.

 For a joint distribution $\mu \in \Delta^{|S|-1}$, let 
\[
	\mu_i(s_i) = \sum_{s_{-i} \in S_{-i}} \mu(s_i,s_{-i})
\] 
be the marginal distribution of player $i$, and for $s_i$ such that $\mu_i(s_i) > 0$ let
\[
	\mu_{-i}^{s_i}(s_{-i}) = \frac{\mu(s_i,s_{-i})}{\mu_i(s_i)}
\]
be the conditional distribution on $S_{-i}$.

If player $i$ observes a signal to play $s_i$ drawn from the joint distribution $\mu$, and if she decides to deviate to a strategy $d_i \in S_i$, then she will face the lottery
\[
	L(\mu,s_i,d_i) := \l\{ \l(\mu_{-i}^{s_i}(s_{-i}), h_i(d_i,s_{-i}) \r)\r\}_{s_{-i} \in S_{-i}}.
\]

\begin{definition}[\cite{keskin2016equilibrium}]
	A joint probability distribution $\mu \in \Delta^{|S|-1}$ is said to be a \emph{CPT correlated equilibrium} of $\Gamma$ if it satisfies the following inequalities for all $i$ and for all $s_i,d_i \in S_i$ such that $\mu_i(s_i) > 0$:
	\begin{equation}\label{eq: CPT_corr_ineq}
		V_i^{r_i(\mu)}(L(\mu,s_i,s_i)) \geq V_i^{r_i(\mu)}(L(\mu,s_i,d_i)). 
	\end{equation}
\end{definition}

For any fixed reference point $r$, since the value function $v^r(\cdot)$ is assumed to be strictly increasing, one can check that 
two outcome profiles $x$ and $y$ 
have equal CPT value under all probability distributions $p$,
i.e. $V^r(p,x) = V^r(p,y)$ for all $p$,
iff $x = y$. 
It then follows that the set $C_{CPT}$ is a proper subset of $\Delta^{|S|-1}$ iff the game is non-trivial.

We now describe the notion of CPT Nash equilibrium as defined by Keskin\footnote{Keskin defines CPT equilibrium assuming $w^+(\cdot) = w^-(\cdot)$. However, the definition can be easily extended to our general setting and the proof of existence goes through without difficulty.} \cite{keskin2016equilibrium}. For a mixed strategy $\mu \in I$, if player $i$ decides to play $s_i$, drawn from the distribution $\mu_i$, then she will face the lottery
\[
	L(\mu_{-i},s_i) := \l\{\l(\mu_{-i}(s_{-i}), h_i(s_i,s_{-i}) \r)\r\}_{s_{-i} \in S_{-i}},
\]
where $\mu_{-i}(s_{-i}) = \prod_{j \neq i} \mu_j(s_j)$ plays the role of $\mu^{s_i}_{-i}(s_{-i})$,
which does not depend on $s_i$.
Suppose player $i$ decides to deviate and play a mixed strategy $\mu'_i$ while the rest of the players continue to play $\mu_{-i}$. Then define
the average CPT value for player $i$ by
\[
	A_i^{CPT}(\mu'_i,\mu_{-i}) = \sum_{s_i \in S_i} \mu'_i(s_i)V_i^{r_i(\mu)}(L(\mu_{-i},s_i)).
\]
The best response of player $i$ to a mixed strategy $\mu \in I$ is defined as
\[
	BR_i^{CPT}(\mu) := \l\{\mu^*_i \in \Delta^{|S_i|-1} | \forall \mu'_i \in \Delta^{|S_i|-1}, A_i^{CPT}(\mu^*_i,\mu_{-i}) \geq A_i^{CPT}(\mu'_i,\mu_{-i}) \r\}.
\]
\begin{definition}[\cite{keskin2016equilibrium}]
	A mixed strategy $\mu^* \in I$, is a CPT Nash equilibrium iff
	\[
		\mu^*_i \in BR_i^{CPT}(\mu^*) \text{ for all } i.
	\]
	We call $\mu^*$ a pure or mixed CPT Nash equilibrium depending on $\mu^*$ being a pure or mixed strategy respectively.
\end{definition}

The set of all CPT correlated equilibria, henceforth denoted by $C_{CPT}$, is no longer guaranteed to be a convex polytope \citep[Example~2]{keskin2016equilibrium}. The set of all CPT Nash equilibria is the intersection of $I$ and $C_{CPT}$ \citep[Proposition~1]{keskin2016equilibrium}) and is non-empty \citep[Theorem~1]{keskin2016equilibrium}. We are interested in studying the geometry of this intersection. It should be noted that the set $C_{CPT}$ depends on the choice of the reference functions $r_i(\mu)$, as does the set of CPT Nash equilibria.

In the case of traditional utility-theoretic equilibria, it has been proved that
\begin{proposition}[\cite{nau2004geometry}]
	In any finite, non-trivial game, the Nash equilibria are on the boundary of the polytope of correlated equilibria when it is viewed as a subset of the smallest affine space containing all joint probability distributions.
\end{proposition}

Since the set of correlated equilibria $C_{EUT}$ is a convex polytope, it is enough to prove that the Nash equilibria lie on one of the faces of $C_{EUT}$ if $C_{EUT}$ is full-dimensional, i.e. has dimension $|S| - 1$, when it is viewed as a subset of the affine space containing $\Delta^{|S|-1}$, and the statement is trivially true if it is not full-dimensional. When the set $C_{EUT}$ is not full-dimensional, it is possible for the Nash equilibria to lie in the relative interior of the set $C_{EUT}$ \citep[Proposition 2]{nau2004geometry}.

We now extend the above proposition for equilibria with CPT preferences. The proof is quite different since in general $C_{CPT}$ is not a convex polytope, as shown in Example~\ref{ex: disconn}, Section~\ref{sec: conn} below \citep[see also][Example~2]{keskin2016equilibrium}.

\begin{proposition}\label{prop: main}
	In any finite, non-trivial game, the CPT Nash equilibria are on the boundary of the set of CPT correlated equilibria set when it is viewed as a subset of the smallest affine space containing all joint probability distributions.
\end{proposition}
We first prove a lemma which in itself is an interesting property of cumulative prospect theoretic preferences. Let $V^r(\cdot)$ denote the CPT value evaluated with respect to a value function $v^r(\cdot)$ and probability weighting functions $w^\pm(\cdot)$ with respect to a reference point $r \in \bbR$. Let $x = (x_1,\dots,x_t)$ and $y = (y_1,\dots,y_t)$ be two outcome profiles and $p = (p_1,\dots,p_t)$ be a  probability distribution. The prospect $(p,x)$ is said to \emph{pointwise dominate} the prospect $(p,y)$ if $x_j \geq y_j$ for all $j$ such that $p_j > 0$. Further, if the inequality $x_j \geq y_j$ holds strictly for at least one $j$ with $p_j > 0$ then the prospect $(p,x)$ is said to \emph{strictly pointwise dominate} the prospect $(p,y)$. Let the \emph{regret} corresponding to choosing $y$ instead of $x$ be denoted by
\begin{equation}\label{eq: regret}
	R^r(p,x,y) := V^r(p,x) - V^r(p,y).
\end{equation}
Prospects $(p,x)$ and $(p,y)$ are said to be \emph{similarly ranked} if there exists a permutation $(a_1,\dots,a_{t'})$ of $T' :=\{j \in \{1,\dots,t\}| p_j > 0\}$ such that
\[
	x_{a_1} \geq \dots \geq x_{a_{t'}} \text{ and } y_{a_1} \geq \dots \geq y_{a_{t'}}.
\]

\begin{lemma}\label{lem: CPT_corr_regret_dir}
In the above setting, suppose the prospects $(p,x)$ and $(p,y)$ satisfy either of the following:
\begin{enumerate}[(i)]
 	\item they are not similarly ranked or,
 	\item neither of them dominates the other,
 \end{enumerate} 
 then there exists a direction $\delta = (\delta_1,\dots,\delta_t)$ with $\sum_{j = 1}^t \delta_j = 0$ and $\delta_j = 0$ for $j \notin T'$ such that
 \begin{equation}\label{eq: reg_dir}
  	R^r(p + \epsilon \delta, x, y) < R^r(p , x, y)
  \end{equation} 
 for all $r \in \bbR$, for all $\epsilon > 0$ such that $p + \epsilon \delta \in \Delta^{t-1}$.
\end{lemma}

\begin{proof}
	We observe that it is enough to prove the claim for the case when $p_j > 0$ for all $1\leq j\leq t$ because if not, then we can let $x',y'$ and $p'$ be respectively the vectors $x,y$ and $p$ restricted to the coordinates in $T'$ and then use the result. WLOG let the ordering be such that $x_1 \geq \dots \geq x_t$. Let $\delta(j_1,j_2)$ correspond to transferring probability from $j_1$ to $j_2$, i.e. for all $1\leq j \leq t$,
	\[
		\delta_j(j_1,j_2) = \begin{cases}
			1 \text{ if } j = j_2,\\
			-1 \text{ if } j = j_1,\\
			0 \text{ otherwise }.
		\end{cases} 
	\]
	Suppose (i) holds. Then there exists $j_1 < j_2$ (and hence $x_{j_1} \geq x_{j_2}$) such that $y_{j_1} < y_{j_2}$. Now, by the strict stochastic dominance property of CPT we have 
	\[
		V^r(p + \epsilon \delta,x) \leq V^r(p,x) \text{ and } V^r(p,y) < V^r(p + \epsilon \delta,y),
	\]
	where $\delta$ denotes $\delta(j_1,j_2)$,
	and  hence (\ref{eq: reg_dir}) follows.
	
	Now suppose $(p,x)$ and $(p,y)$ are similarly ranked. WLOG let the ordering be such that $x_1 \geq \dots \geq x_t$ and $y_1 \geq \dots \geq y_t$. Suppose (ii) holds. Then there exist $j_1,j_2$ such that $x_{j_1} > y_{j_1}$ and $x_{j_2} < y_{j_2}$. In fact, one can find $j_1, j_2$ such that $x_{j_1} > y_{j_1},  x_{j} = y_j$ for all $j$ between $j_1$ and $j_2$, and $x_{j_2} < y_{j_2}$. Depending on the order of $j_1$ and $j_2$ we have the following two cases (note $j_1 \neq j_2$):

	\noindent
	\emph{Case} 1 ($j_1 < j_2$): Then we have the ordering $x_{j_2} < y_{j_2} \leq y_{j_1} < x_{j_1}$. Let $\delta = \delta(j_1,j_2)$. Then it follows from the strict monotonicity of the functions $w_i^\pm(\cdot)$ and the definition of decision weights that
    \begin{align*}
    \pi^+_{j_1}(p + \epsilon \delta) - \pi^+_{j_1}(p) &< 0,\\
    \pi^+_{j_2}(p + \epsilon \delta) - \pi^+_{j_2}(p) &> 0,\\
    \pi^-_{j_1}(p + \epsilon \delta) - \pi^-_{j_1}(p) &< 0,\\
    \pi^-_{j_2}(p + \epsilon \delta) - \pi^-_{j_2}(p) &> 0.
    \end{align*}
(We suppress the dependence of $\pi^\pm_j(p,a)$ on the permutation $a$ since we have assumed $x$ and $y$ to be ordered.) Depending on the position of the reference point $r$, we have the following subcases:

	\noindent
	\emph{Subcase} 1a ($r \leq x_{j_2}$):
	\begin{align*}
		&[V^r(p + \epsilon \delta,x) - V^r(p + \epsilon \delta,y)] - [V^r(p, x) -V^r(p, y)]\\
		&= [\pi^+_{j_1}(p + \epsilon \delta) - \pi^+_{j_1}(p)][v^r(x_{j_1}) - v^r(y_{j_1})]\\
		 &+ [\pi^+_{j_2}(p + \epsilon \delta) - \pi^+_{j_2}(p)][v^r(x_{j_2}) - v^r(y_{j_2})], 
	\end{align*}
	because $\pi^+_{j}(p + \epsilon \delta) = \pi^+_{j}(p)$ for all $j \notin \{j_1,\dots,j_2\}$ and $v^r(x_j) = v^r(y_j)$ for all $j_1 < j < j_2$.  Since $v^r(x_{j_1}) - v^r(y_{j_1}) > 0$ and $v^r(x_{j_2}) - v^r(y_{j_2}) < 0$ we get (\ref{eq: reg_dir}).
	
	\noindent
	\emph{Subcase} 1b ($x_{j_2} < r \leq y_{j_2}$):
	\begin{align*}
		&[V^r(p + \epsilon \delta,x) - V^r(p + \epsilon \delta,y)]- [V^r(p, x) -V^r(p, y)]\\
		&= [\pi^+_{j_1}(p + \epsilon \delta) - \pi^+_{j_1}(p)][v^r(x_{j_1}) - v^r(y_{j_1})] + [\pi^-_{j_2}(p + \epsilon \delta) - \pi^-_{j_2}(p)]v^r(x_{j_2}) \\
		& \quad \quad - [\pi^+_{j_2}(p + \epsilon \delta) - \pi^+_{j_2}(p)]  v^r(y_{j_2}).
	\end{align*}
	Now $v^r(x_{j_1}) - v^r(y_{j_1}) > 0$,  $v^r(x_{j_2}) < 0$, $v^r(y_{j_2}) > 0$ and the result follows.
	
	\noindent
	\emph{Subcase} 1c ($y_{j_2} < r \leq y_{j_1}$):
	\begin{align*}
		&[V^r(p + \epsilon \delta,x) - V^r(p + \epsilon \delta,y)]- [V^r(p, x) - V^r(p, y)]\\
		&= [\pi^+_{j_1}(p + \epsilon \delta) - \pi^+_{j_1}(p)][v^r(x_{j_1}) - v^r(y_{j_1})] \\
		&+ [\pi^-_{j_2}(p + \epsilon \delta) - \pi^-_{j_2}(p)][v^r(x_{j_2}) -  v^r(y_{j_2})].
	\end{align*}
	Now $v^r(x_{j_1}) - v^r(y_{j_1}) > 0$, $v^r(x_{j_2}) - v^r(y_{j_2}) < 0$ and the result follows.

	\noindent
	\emph{Subcase} 1d ($y_{j_1} < r \leq x_{j_1}$):
	\begin{align*}
		&[V^r(p + \epsilon \delta,x) - V^r(p + \epsilon \delta,y)]- [V^r(p, x) -V^r(p, y)]\\
		&= [\pi^+_{j_1}(p + \epsilon \delta) - \pi^+_{j_1}(p)]v^r(x_{j_1}) - [\pi^-_{j_1}(p + \epsilon \delta) - \pi^-_{j_1}(p)]v^r(y_{j_1}) \\
		& \quad \quad + [\pi^-_{j_2}(p + \epsilon \delta) - \pi^-_{j_2}(p)][v^r(x_{j_2}) -  v^r(y_{j_2})].
	\end{align*}
	Now $v^r(x_{j_1}) > 0$, $v^r(y_{j_1}) < 0$,  $v^r(x_{j_2}) - v^r(y_{j_2}) < 0$ and the result follows.

	\noindent
	\emph{Subcase} 1e ($x_{j_i} < r$):
	\begin{align*}
		&[V^r(p + \epsilon \delta,x) - V^r(p + \epsilon \delta,y)]- [V^r(p, x) -V^r(p, y)]\\
		&= [\pi^-_{j_1}(p + \epsilon \delta) - \pi^-_{j_1}(p)][v^r(x_{j_1}) - v^r(y_{j_1})] \\
		&+ [\pi^-_{j_2}(p + \epsilon \delta) - \pi^-_{j_2}(p)][v^r(x_{j_2}) - v^r(y_{j_2})]. 
	\end{align*}
	Now $v^r(x_{j_1}) - v^r(y_{j_1}) > 0$,  $v^r(x_{j_2}) - v^r(y_{j_2}) < 0$ and the result follows.

	Case 2 ($j_1 > j_2$) implies the order $y_{j_1} < x_{j_1} \leq x_{j_2} < y_{j_2}$. Taking $\delta = \delta(j_2,j_1)$, each of the subcases depending on the position of the reference point can be handled as in case 1.
\end{proof}
\begin{remark}\label{rem: lemma_nodepend}
	The vector $\delta$ used in the proof above depends only on the prospects $(p,x)$ and $(p,y)$ and not on the reference point $r$. In fact, it depends only on the order structure of the vectors $x$ and $y$ and not on the probability distribution vector $p$ as long as $p_j > 0$ for all $1 \leq j \leq t$. Also, the range of $\epsilon$ for which the claim holds depends only on the prospects $(p,x)$ and $(p,y)$ and not on the reference point $r$. Lemma~\ref{lem: CPT_corr_regret_dir} can be extended to more general CPT settings as in \citet{chateauneuf1999axiomatization}, where the outcome space is assumed to be a connected topological space instead of monetary outcomes in $\bbR$.
\end{remark}


\begin{proof}[Proof of proposition~\ref{prop: main}]
	If a CPT Nash equilibrium $\hat \mu$ is not completely mixed, i.e. there is a player $i$ and a strategy $s_i \in S_i$, such that $\hat \mu_i(s_i) = 0$, then $\hat \mu$ assigns zero probability to one or more joint strategies and hence lies on the boundary of $\Delta^{|S|-1}$ and thus also on the boundary of $C_{CPT}$. 

	Suppose now that $\hat \mu \in I \cap C_{CPT}$ is completely mixed. Then the inequalities~(\ref{eq: CPT_corr_ineq}) hold for all $i$ and for all $s_i,d_i \in S_i$. In particular, for any pair $s_i,d_i \in S_i$ we have
	\begin{align*}
		V_i^{r_i(\hat \mu)}&\l(\l\{\l( \hat \mu_{-i}^{s_i}(s_{-i}), h_i(s_i,s_{-i})\r) \r\}_{s_{-i} \in S_{-i}}\r) \\
        &\geq V_i^{r_i(\hat \mu)}\l(\l\{\l( \hat  \mu_{-i}^{s_i}(s_{-i}), h_i(d_i,s_{-i}) \r) \r\}_{s_{-i} \in S_{-i}}\r),\\
		V_i^{r_i(\hat \mu)}&\l(\l\{\l( \hat \mu_{-i}^{d_i}(s_{-i}), h_i(d_i,s_{-i})\r) \r\}_{s_{-i} \in S_{-i}}\r) \\
        &\geq V_i^{r_i(\hat \mu)}\l(\l\{\l( \hat \mu_{-i}^{d_i}(s_{-i}), h_i(s_i,s_{-i}) \r)\r\}_{s_{-i} \in S_{-i}}\r).
	\end{align*}
	However, since $\hat \mu \in I$, we have $\hat \mu_{-i} := \hat \mu_{-i}^{s_i} = \hat \mu_{-i}^{d_i}$ and hence the above inequalities are in fact equalities. The same is true for all the inequalities~(\ref{eq: CPT_corr_ineq}).

	Since the game is non-trivial, there exist $i \in N$ and $s_i,d_i \in S_i$ such that $h_i(s_i,s_{-i}) \neq h_i(d_i,s_{-i})$ for some $s_{-i} \in S_{-i}$. Consider the inequality in (\ref{eq: CPT_corr_ineq}) corresponding to such an $(i, s_i,d_i)$. Fix a one to one correspondence between the numbers $\{1,\dots,t\}$ and the joint strategies $\{s_{-i} \in S_{-i}\}$ (here $t = |S_{-i}|$). Let 
	\[
		x := (x_1, \dots, x_t) = (h_i(s_i,s_{-i}) )_{s_{-i} \in S_{-i}},
	\]
	 and 
	 \[
	 	y := (y_1,\dots,y_t) = (h_i(d_i,s_{-i}))_{s_{-i} \in S_{-i}}.
	 \]
	 Since $\hat \mu$ is completely mixed, $\hat \mu_i(s_i) > 0$. Let 
	  \[
	  	p = (p_1,\dots,p_t) = (\hat{\mu}_{-i}(s_{-i}))_{s_{-i} \in S_{-i}}
	  \]
	   be the conditional probability distribution on $S_{-i}$. 

	If either profile $(p,x)$ or $(p,y)$ pointwise dominated the other then the pointwise dominance would be strict since $x$ and $y$ are distinct and $p_j > 0$ for all $1 \leq j \leq t$. By the strict monotonicity property of CPT, we would get $V^{r_{i}(\hat \mu)}(p,x) \neq V^{r_i(\hat \mu)}(p,y)$ contrary to our assumption. Thus condition (ii) in Lemma~\ref{lem: CPT_corr_regret_dir} is satisfied and there exists a direction vector $\delta = (\delta_1,\dots,\delta_t)$ with $\sum_{j=1}^t \delta_j = 0$ such that $V_i^{r_i}(p + \epsilon \delta,x) < V_i^{r_i}(p + \epsilon \delta,y)$ for all $r_i \in \bbR$, for all $\epsilon>0$ such that $p + \epsilon \delta \in \Delta^{t-1}$. Note that the vector $\delta$ and the range of $\epsilon$ does not depend on the reference point $r_i$ (see remark~\ref{rem: lemma_nodepend}). Consider the joint probability distribution $\bar \mu$ given by
	\begin{align*}
		\bar \mu(c_i,s_{-i}) = \begin{cases}
			\hat \mu_{i}(s_i)(p_j+\epsilon \delta_j) \text{ if } c_i = s_i \text{ and $j$ corresponds to $s_{-i}$ },\\
			\hat \mu(c_i,s_{-i}) \text{ otherwise.}
		\end{cases}
	\end{align*}

Let $R_i^r(\cdot)$ denote the regret corresponding to player $i$, evaluated using her value function and probability weighting functions. This should be thought of as defined for any pair of outcome profiles $x$ and $y$  on $S_{-i}$ with a given probability distribution $p$ on $S_{-i}$, as in equation~(\ref{eq: regret}), with $V^r$ being replaced by $V_i^r$ and defined as in equation~(\ref{eq: CPT_value_discrete}), using the value function $v_i^r$ and the weighting functions $w_i^{\pm}$. Since $\hat \mu \in I \cap C_{CPT}$,
\[
	R_i^{r_i(\hat \mu)}(\hat \mu_{-i},x,y) = V_i^{r_i(\hat \mu)}(p,x) -  V_i^{r_i(\hat \mu)}(p,y) = 0,
\]
and
\[
	R_i^{r_i(\hat \mu)}(\hat \mu_{-i},y,x) = V_i^{r_i(\hat \mu)}(p,y) -  V_i^{r_i(\hat \mu)}(p,x) = 0.
\]
respectively.
Now if
\[
	R_i^{r_i(\bar \mu)}(\hat \mu_{-i},x,y) \leq R_i^{r_i(\hat \mu)}(\hat \mu_{-i},x,y)
\]
then from the choice of $\bar \mu$
\[
	R_i^{r_i(\bar \mu)}(\bar \mu_{-i}^{s_i},x,y)  = R_i^{r_i(\bar \mu)}(p + \epsilon \delta,x,y) < R_i^{r_i(\bar \mu)}(p,x,y)  = R_i^{r_i(\bar \mu)}(\hat \mu_{-i},x,y) \leq 0.
\]
On the other hand, if
\[
	R_i^{r_i(\bar \mu)}(\hat \mu_{-i},x,y) > R_i^{r_i(\hat \mu)}(\hat \mu_{-i},x,y) = 0
\]
then
\[
	R_i^{r_i(\bar \mu)}(\bar \mu_{-i}^{d_i},y,x) = R_i^{r_i(\bar \mu)}(\hat \mu_{-i},y,x) = - R_i^{r_i(\bar \mu)}(\hat \mu_{-i},x,y) < 0.
\]
Thus either of the inequalities in (\ref{eq: CPT_corr_ineq}) corresponding to deviation from $s_i$ to $d_i$ or $d_i$ to $s_i$ is violated by the joint distribution $\bar \mu$. Thus, for any neighborhood $N$ of $\hat \mu$, $\bar \mu$ belongs to $N$ for sufficiently small $\epsilon$ and $\bar \mu \notin C_{CPT}$. Thus $\hat \mu$ lies on the boundary of $C_{CPT}$.
\end{proof}


\section[2x2games]{$2\times 2$ games}\label{sec: 2x2games}
	 For a game $\Gamma$, the set $C_{CPT}$, in general, need not be convex (example 2 in \cite{keskin2016equilibrium}). In this section we will see that, in the special case of a $2\times2$ game with players having a fixed reference point independent of the underlying probability distribution, $C_{CPT}$ is a convex polytope.

	 Consider a $2$ player game $\Gamma$ with $N = \{1,2\}$ and $S_1 = S_2 = \{0,1\}$. With player $1$ as the row player and player $2$ as the column player and $\{a_{ij}, b_{ij}\}_{i,j \in \{0,1\} }$ representing payoffs for player $1$ and $2$ respectively, let the payoff matrix be as shown in figure~\ref{fig: 2x2game}. 
	 Here, the real numbers $a_{ij}$ and $b_{ij}$ should be thought of as 
	 outcomes in the terminology of cumulative prospect theory, but we will call them payoffs in this section. Let $\mu = \{\mu_{00},\mu_{01},\mu_{10},\mu_{11} \} \in \Delta^3$ be a joint probability distribution assigning probabilities to joint strategies as represented by the matrix in figure \ref{fig: 2x2game}. Let $r_1$ and $r_2$ be the fixed reference points (independent of the joint probability distribution $\mu$) for players $1$ and $2$ respectively.

	\begin{figure}[h]
    \parbox{.45\linewidth}{
	 \centering
	 \begin{tabular}{c | c | c |}
	 	\multicolumn{1}{c}{} & \multicolumn{1}{c}{0} & \multicolumn{1}{c}{1} \\
	 \cline{2-3}
	 0	& $a_{00},b_{00}$ & $a_{01},b_{01}$\\
	 \cline{2-3}
	 1	& $a_{10},b_{10}$ & $a_{11},b_{11}$\\
	 \cline{2-3}
	 \end{tabular}
     }
	\hfill
    \parbox{.45\linewidth}{
	 \centering
	 \begin{tabular}{c | c | c |}
	 	\multicolumn{1}{c}{} & \multicolumn{1}{c}{0} & \multicolumn{1}{c}{1} \\
	 \cline{2-3}
	 0	& $\mu_{00}$ & $\mu_{01}$\\
	 \cline{2-3}
	 1	& $\mu_{10}$ & $\mu_{11}$\\
	 \cline{2-3}
	 \end{tabular}
     }
     \caption{Payoff matrix (left) and joint probability matrix (right) of a $2\times 2$ game}\label{fig: 2x2game}
	 \end{figure}


\begin{proposition}\label{prop: 2x2convex}
	 	For the above $2 \times 2$ game, the set $C_{CPT}$ is a convex polytope.
	 \end{proposition}
	 \begin{proof}
	 The condition for $\mu \in C_{CPT}$  corresponding to the row player deviating from strategy $0$ to strategy $1$ in (\ref{eq: CPT_corr_ineq}) is:
	 \begin{align}\label{eq: corr_2_2_cond_1}
	  	\text{if $\mu_{00} + \mu_{01} > 0$ then } R_1^{r_1}(p^1,x,y) \geq 0,
	  \end{align} 
	  where $p^1 = (p^1_0,p^1_1)$, $p^1_0 = \frac{\mu_{00} }{\mu_{00} + \mu_{01}}, \quad p^1_1 =  \frac{\mu_{01}}{\mu_{00} + \mu_{01}}, x = (a_{00},a_{01}), y = (a_{10},a_{11})$. Let $C_1$ denote the set of all $\mu \in \Delta^3$ satisfying condition~(\ref{eq: corr_2_2_cond_1}). We have:
	  \begin{enumerate}[(i)]
	  	\item if $a_{00} \geq a_{10}$ and $a_{01} \geq a_{11}$, then $C_1 = \Delta^3$;
	  	\item if $a_{00} < a_{10}$ and $a_{01} = a_{11}$ (resp. $a_{00} = a_{10}$ and $a_{01} < a_{11}$), then $C_1 = \{\mu \in \Delta^3 | \mu_{00} = 0\}$ (resp. $C_1 = \{\mu \in \Delta^3 | \mu_{01} = 0\}$);
	  	\item If $a_{00} < a_{10}$ and $a_{01} < a_{11}$, then $C_1 = \{\mu \in \Delta^3 | \mu_{00} = 0, \mu_{01} = 0\}$;
	  	\item if $a_{00} < a_{10}$ and $a_{01} > a_{11}$ (resp. $a_{00} > a_{10}$ and $a_{01} < a_{11}$), then from lemma \ref{lem: CPT_corr_regret_dir}, $R_1^{r_1}(p^1,x,y)$ is strictly monotonic as a function of $p_0^1$ $(= 1 - p^1_1)$ on the interval $(0,1)$, 
	  	\begin{align*}
	  		&R_1^{r_1}((0,1),x,y) > 0 > R_1^{r_1}((1,0),x,y)\\
	  		&(\text{resp. $R_1^{r_1}((0,1),x,y) < 0 < R_1^{r_1}((1,0),x,y))$} ),
	  	\end{align*}
	  	and hence the inequality in condition~(\ref{eq: corr_2_2_cond_1}) holds iff $p^1_0 \leq q_0$ (resp. $p^1_1 \leq q_1$) for a certain $q_0 \in (0,1)$ (resp. $q_1 \in (0,1)$) depending on the payoffs $a_{00},a_{01},a_{10}$ and $a_{11}$, the value function $v_1^r(\cdot)$, and the probability weight functions $w_1^\pm(\cdot)$. Thus $C_1 = \{\mu \in \Delta^3 | \alpha_0 \mu_{00} \leq \mu_{01}\}$ with $\alpha_0 = \frac{1-q_0}{q_0}$ (resp. $C_1 = \{\mu \in \Delta^3 | \alpha_1 \mu_{00} \geq \mu_{01}\}$ with $\alpha_1 = \frac{q_1}{1-q_1}$ ).
	  \end{enumerate}
	
	In each case, $C_1$ is a convex polytope. Similarly, the other three conditions in (\ref{eq: CPT_corr_ineq}), corresponding to the row player deviating from strategy $1$ to strategy $0$, the column player deviating from strategy $0$ to strategy $1$, and the column player deviating from strategy $1$ to strategy $0$, give rise to convex polytopes $C_2,C_3$ and $C_4$ respectively. The set $C_{CPT}$ is the (non-empty) intersection of these convex polytopes and hence is itself a convex polytope.
	\end{proof}


\begin{definition}
	For an $n$ player game $\Gamma = (N,(S_i)_{i\in N}, (h_i)_{i \in N})$, let $s_i,d_i \in S_i$ be two strategies corresponding to player $i$.
	\begin{itemize}
	 	\item Strategies $s_i$ and $d_i$ are said to be equivalent if player $i$ is indifferent in choosing between $s_i$ and $d_i$ no matter what the other players do.
	 	\item Strategy $s_i$ is said to be weakly dominated by strategy $d_i$ if there exists at least one strategy profile of the opponents for which choosing $d_i$ is better than choosing $s_i$, and for all strategy profiles of the opponents choosing $d_i$ is at least as good as choosing $s_i$.
	 	\item Strategy $s_i$ is said to be strictly dominated by strategy $d_i$ if, for every strategy profile of the opponents, choosing $d_i$ is better than choosing $s_i$.
	 \end{itemize} 
\end{definition}

Note that a strictly dominated strategy is also a weakly dominated strategy. 

As observed in section~\ref{sec: main}, two outcome profiles $x$ and $y$ 
are equivalent under all probability distributions $p$ iff $x = y$. Thus, as under EUT, for players with CPT preferences we have the following:
\begin{itemize}
	\item Strategy $s_i$ is equivalent to strategy $d_i$ iff
	\[
		h_i(s_i,s_{-i}) = h_i(d_i,s_{-i}) \quad \forall s_{-i} \in S_{-i}.
	\]
	\item Strategy $s_i$ is weakly dominated by strategy $d_i$ iff
	\[
		h_i(s_i,s_{-i}) \leq h_i(d_i,s_{-i}) \quad \forall s_{-i} \in S_{-i},
	\]
	where strict inequality holds for at least one $s_{-i} \in S_{-i}$.
	\item Strategy $s_i$ is strictly dominated by strategy $d_i$ iff
	\[
		h_i(s_i,s_{-i}) < h_i(d_i,s_{-i}) \quad \forall s_{-i} \in S_{-i}.
	\]
\end{itemize}
We now look at the convex polytope $C_{CPT}$ for a $2 \times 2$ game in more detail. We first discuss $2 \times 2$ games with no equivalent or weakly dominated strategies. Let $G^0$ denote the set of all such games.
For any game $\Gamma \in G^0$, the relation amongst the payoffs for all the four conditions corresponding to $C_1,C_2,C_3$ and $C_4$ are as in case (iv) above. Further, the conditions corresponding to the row player deviating from strategy $0$ to strategy $1$, and vice versa are
	\begin{equation}\label{eq: cond_1}
		\text{if $\mu_{00} + \mu_{01} > 0$ then } V_1^{r_1}(p^1,x) \geq V_1^{r_1}(p^1,y);
    \end{equation}
and
	\begin{equation}\label{eq: cond_2}
     \text{if $\mu_{10} + \mu_{11} > 0$ then } V_1^{r_1}(p^2,x) \leq V_1^{r_1}(p^2,y);
	\end{equation}
	respectively, where $p^1$ is as in proposition~\ref{prop: 2x2convex} and $p^2 = (p^2_0,p^2_1)$, $p^2_0 = \frac{\mu_{10}}{\mu_{10} + \mu_{11}}$ and $p^2_1 = \frac{\mu_{11}}{\mu_{10} + \mu_{11}}$. Now there exists a $q_0 \in (0,1)$ (or a $q_1 \in (0,1)$) such that inequality~(\ref{eq: cond_1}) holds for all $p^1_0 \leq q_0$ (resp. $p^1_1 \leq q_1$) and inequality~(\ref{eq: cond_2}) holds for all $p^2_0 \geq q_0$ (resp. $p^2_1 \geq q_1$). Thus if $C_1 = \{\mu \in \Delta^3 | \alpha_0 \mu_{00} \leq \mu_{01}\}$ (resp. $C_1 = \{\mu \in \Delta^3 | \alpha_1 \mu_{00} \geq \mu_{01}\}$), then $C_2 = \{\mu \in \Delta^3 | \alpha_0 \mu_{10} \geq \mu_{11}\}$ (resp. $C_2 = \{\mu \in \Delta^3 | \alpha_1 \mu_{10} \leq \mu_{11}\}$). Similarly for player $2$. Thus, depending on the relation amongst the payoffs, the conditions~(\ref{eq: CPT_corr_ineq}) take one of the following forms:
	\begin{enumerate}[(I)]
		\item if $a_{00} > a_{10}, a_{01} < a_{11}, b_{00} > b_{01}, b_{10} < b_{11}$ then
		\[
			\alpha \mu_{00} \geq \mu_{01},\;\alpha \mu_{10} \leq \mu_{11},\; \beta \mu_{00} \geq  \mu_{10},\; \beta \mu_{01} \leq \mu_{11};
		\]
		\item if $a_{00} < a_{10}, a_{01} > a_{11}, b_{00} > b_{01}, b_{10} < b_{11}$ then
		$$\alpha \mu_{00} \leq \mu_{01},\;\alpha \mu_{10} \geq \mu_{11},\; \beta \mu_{00} \geq  \mu_{10},\; \beta \mu_{01} \leq \mu_{11};$$
		\item if $a_{00} > a_{10}, a_{01} < a_{11}, b_{00} < b_{01}, b_{10} > b_{11}$ then
		$$\alpha \mu_{00} \geq \mu_{01},\;\alpha \mu_{10} \leq \mu_{11},\; \beta \mu_{00} \leq  \mu_{10},\; \beta \mu_{01} \geq \mu_{11};$$
		\item if $a_{00} < a_{10}, a_{01} > a_{11}, b_{00} < b_{01}, b_{10} > b_{11}$ then
		$$\alpha \mu_{00} \leq \mu_{01},\;\alpha \mu_{10} \geq \mu_{11},\; \beta \mu_{00} \leq  \mu_{10},\; \beta \mu_{01} \geq \mu_{11};$$
	\end{enumerate}
for some $\alpha, \beta > 0$. Thus every $2\times 2$ game with no equivalent or weakly dominated strategies can be classified into one of the above four types depending on the relations amongst its payoffs.

We consider the canonical $2\times 2$ games $\gamma_l(\alpha,\beta)$ for $l \in \{I,II,III,IV\}$ with $\alpha,\beta > 0$ as shown in figure~\ref{fig: canon_games}.
\begin{figure}
\parbox{.45\linewidth}{
\centering
\begin{tabular}{c | c | c |}
	 	\multicolumn{1}{c}{} & \multicolumn{1}{c}{0} & \multicolumn{1}{c}{1} \\
	 \cline{2-3}
	 0	& $\alpha,\beta$ & $0,0$\\
	 \cline{2-3}
	 1	& $0,0$ & $1,1$\\
	 \cline{2-3}
	 \end{tabular}

\caption*{$\gamma_{I}(\alpha,\beta)$}
}
\hfill
\parbox{.45\linewidth}{
\centering
\begin{tabular}{c | c | c |}
	 	\multicolumn{1}{c}{} & \multicolumn{1}{c}{0} & \multicolumn{1}{c}{1} \\
	 \cline{2-3}
	 0	& $-\alpha,\beta$ & $0,0$\\
	 \cline{2-3}
	 1	& $0,0$ & $-1,1$\\
	 \cline{2-3}
	 \end{tabular} 

\caption*{$\gamma_{II}(\alpha,\beta)$}
}
\vskip\baselineskip
\parbox{.45\linewidth}{
\centering
\begin{tabular}{c | c | c |}
	 	\multicolumn{1}{c}{} & \multicolumn{1}{c}{0} & \multicolumn{1}{c}{1} \\
	 \cline{2-3}
	 0	& $\alpha,-\beta$ & $0,0$\\
	 \cline{2-3}
	 1	& $0,0$ & $1,-1$\\
	 \cline{2-3}
	 \end{tabular} 

\caption*{$\gamma_{III}(\alpha,\beta)$}
}
\hfill
\parbox{.45\linewidth}{
\centering
\begin{tabular}{c | c | c |}
	 	\multicolumn{1}{c}{} & \multicolumn{1}{c}{0} & \multicolumn{1}{c}{1} \\
	 \cline{2-3}
	 0	& $-\alpha,-\beta$ & $0,0$\\
	 \cline{2-3}
	 1	& $0,0$ & $-1,-1$\\
	 \cline{2-3}
	 \end{tabular} 

\caption*{$\gamma_{IV}(\alpha,\beta)$}
}
\caption{Canonical $2\times 2$ games}\label{fig: canon_games}
\end{figure}
One can check that the set $C_{EUT}$ for each of these games is given by the corresponding inequalities above. 

As in the paper by \citet{calvo2006set}, based on the type of inequalities satisfied, we classify all $2 \times 2$ games, with no equivalent or weakly dominated strategies, into three types:
	\begin{itemize}
		\item coordination games if the inequalities take form (I),
		\item anti-coordination games if the inequalities take form (IV) and,
		\item competitive games if the inequalities take either form (II) or form (III).
	\end{itemize}
Since the inequalities above completely characterize the set $C_{CPT}$, it is enough to find the set $C_{EUT}$ for each of the canonical games. For case (II), we have
\[
	\alpha \mu_{00} \leq \mu_{01} \leq \frac{\mu_{11}}{\beta} \text{ and
 } \mu_{11} \leq \alpha \mu_{10} \leq \beta \alpha \mu_{00}. 
 \]
 Thus all inequalities must be satisfied with equality and we get
 \begin{align}
 	\mu_{00} &= \frac{1}{(1+\alpha)(1+\beta)}, \mu_{01} = \frac{\alpha}{(1+\alpha)(1+\beta)},\nonumber\\
 	 \mu_{10} &= \frac{\beta}{(1+\alpha)(1+\beta)}, \mu_{11} = \frac{\alpha \beta}{(1+\alpha)(1+\beta)}.
 \end{align}
Case (III) is similar. Thus, for competitive games, the set $C_{EUT}$ is reduced to a single point, which is also the unique mixed Nash equilibrium. 
For coordination games, the set $C_{EUT}$ is a convex polytope with five vertices \cite{calvo2006set} (figure \ref{fig: vertices}). 
\begin{figure}
	\centering
	\begin{tabular}{c c c c c}
		\hline
		$\mu$ & $\mu_{00}$ & $\mu_{01}$ & $\mu_{10}$ & $\mu_{11}$ \\
		\hline
		$\mu^*_A(\alpha,\beta)$ & $1$& $0$ & $0$ & $0$\\
		$\mu^*_B(\alpha,\beta)$ & $0$& $0$ & $0$ & $1$\\
		$\mu^*_C(\alpha,\beta)$ & $\frac{1}{(1+\alpha)(1 + \beta)}$& $\frac{\alpha}{(1+\alpha)(1 + \beta)}$ & $\frac{\beta}{(1+\alpha)(1 + \beta)}$ & $\frac{\alpha \beta}{(1+\alpha)(1 + \beta)}$\\
		$\mu^*_D(\alpha,\beta)$ & $\frac{1}{1 +\beta + \alpha \beta}$& $0$ & $\frac{\beta}{1 +\beta + \alpha \beta}$ & $\frac{\alpha \beta}{1 +\beta + \alpha \beta}$\\
		$\mu^*_E(\alpha,\beta)$ & $\frac{1}{1 +\alpha + \alpha \beta}$& $\frac{\alpha}{1 +\alpha + \alpha \beta}$ & $0$ & $\frac{\alpha \beta}{1 +\alpha + \alpha \beta}$\\
		\hline
	\end{tabular}
	\caption{Vertices of the convex polytope $C_{EUT}$ for $\gamma_{I}(\alpha,\beta)$ }\label{fig: vertices}
\end{figure}
It intersects the set $I$ at the three vertices $\mu^*_A(\alpha,\beta), \mu^*_B(\alpha,\beta)$ and $\mu^*_C(\alpha,\beta)$ of which the first two are pure Nash equilibria. From the set of inequalities corresponding to cases (I) and (IV) we can see that the joint distribution $\mu = (\mu_{00},\mu_{01},\mu_{10},\mu_{11})$ belongs to $C_{EUT}$ of $\gamma_I(\alpha,\beta)$ iff $\tau(\mu) := (\mu_{10},\mu_{11},\mu_{00},\mu_{01})$ belongs to $C_{EUT}$ of $\gamma_{IV}(\alpha,1/\beta)$ \cite{calvo2006set}. 
Thus, for anti-coordination games, the set $C_{EUT}$ is again a convex polytope with five vertices and it intersects $I$ at three of its vertices with two of them pure Nash equilibria. The vertices can be found from  figure~\ref{fig: vertices}, using the transformation $\tau$, and replacing $\beta$ by $1/\beta$. Since $C_{CPT}$ is determined by the same set of inequalities in $\alpha$ and $\beta$, all the results carry over to $2 \times 2$ games with CPT preferences. In particular, the set $C_{CPT}$ is determined by $\alpha$ and $\beta$. Since the set of CPT Nash equilibria (pure and mixed) is given by the intersection of $I$ and $C_{CPT}$, we have a unique mixed CPT Nash equilibrium and no pure CPT Nash equilibria for competitive games, and one mixed and two pure CPT Nash equilibria for coordination and anti-coordination games.

Consider now a $2 \times 2$ game with at least one strictly dominated strategy. This corresponds to case (i) above with both inequalities strict (i.e. $a_{00} > a_{10}$ and $a_{01} > a_{11}$) or case (iii). Strictly dominated strategies cannot be used with positive probability in any correlated equilibrium of that game. It is easy then to compute the set $C_{CPT}$ for such a game by eliminating the strictly dominated strategies. Suppose strategy $1$ is strictly dominated by strategy $0$. Thus, $a_{00} > a_{10}$ and $a_{01} > a_{11}$. If $b_{00} > b_{01}$ then $C_{CPT} = \{\mu \in \Delta^3 | \mu_{01} = \mu_{10} = \mu_{11} = 0\}$ is a point. If $b_{00} < b_{01}$ then $C_{CPT} = \{\mu \in \Delta^3 | \mu_{00} = \mu_{10} = \mu_{11} = 0\}$ is a point. If $b_{00} = b_{01}$ then $C_{CPT} = \{\mu \in \Delta^3 | \mu_{10} = \mu_{11} = 0\}$ is a line segment. In each case $C_{CPT}$ is containedin $I$. The case when strategy $0$ is strictly dominated by strategy $1$ is similar.

Let $G^1$ denote the set of all $2 \times 2$ games with at least one weakly dominated strategy but no equivalent or strictly dominated strategy. If player 1 has a weakly dominated strategy then this corresponds to case (i) above with one equality and one strict inequality (i.e. $a_{00} > a_{10}$ and $a_{01} = a_{11}$, or $a_{00} = a_{10}$ and $a_{01} > a_{11}$), or case (ii). The set of all $2 \times 2$ games, each game characterized by its payoff matrix, forms an $8$-dimensional Euclidean vector space. Every game $\Gamma \in G^1$ can be seen as a limit of games in $G^0$ of a unique type. For example, a game $\Gamma \in G^1$ with payoffs satisfying $a_{00} < a_{10}, a_{01} = a_{11}, b_{00} > b_{01}, b_{10} < b_{11}$ is the limit, as $\epsilon \downarrow 0$, of games $\Gamma_{\epsilon} \in G^0$ with payoffs same as that of $\Gamma$ except $a_{01}$ replaced by $a_{01} + \epsilon$.
 Each of the games $\Gamma_{\epsilon}$ is of type (II) for sufficiently small $\epsilon > 0$. Further, if $\gamma_{II}(\alpha_\epsilon,\beta_{\epsilon})$ are the canonical games corresponding to $\Gamma_{\epsilon}$ then $\alpha_{\epsilon} \to \infty$ and $\beta_{\epsilon} = \beta$ for some fixed $\beta$. Using this observation we classify every game $\Gamma \in G^1$ into four types $l \in \{I,II,III,IV\}$ and each of these $4$ types into eight further subtypes depending on the limit of $\alpha_{\epsilon}$ and $\beta_{\epsilon}$ going either to $0, \infty$ or some real number in $(0,\infty)$ with at least one of them tending to $0$ or $\infty$. The set $C_{CPT}$ for games of the type (I) and (II) are given in figures~\ref{fig: dom_case1} and \ref{fig: dom_case2} respectively. The set $C_{CPT}$ for types (IV) and (III) can be found using the transformation $\tau$ and replacing $\beta$ by $1/\beta$ (with the convention $1/0 = \infty$ and $1/\infty = 0$) from figures \ref{fig: dom_case1} and \ref{fig: dom_case2} respectively. 

 \begin{figure}[t]
 \centering
  	\begin{tabular}{c | c | c |c |}
  	    \multicolumn{1}{c}{}& \multicolumn{1}{c}{$\beta = 0$} & \multicolumn{1}{c}{$0 < \beta < \infty$} & \multicolumn{1}{c}{$\beta = \infty$}\\
  	    \cline{2-4}
  	 $\alpha = 0$ & $\mu_{01} =  0$ & $\mu_{01} = 0$ & $\mu_{01} = 0$ \\
  	 	& $\mu_{10} = 0$ & $\beta \mu_{00} \geq \mu_{10}$ &  \\
  	\cline{2-4}
  	 $0 < \alpha < \infty$ & $\mu_{10} = 0$ & $-$ & $\mu_{01} = 0$ \\
  	 	& $\alpha \mu_{00} \geq \mu_{01}$ &  & $\alpha \mu_{10} \leq \mu_{11}$\\
  	 \cline{2-4}
  	 $\alpha = \infty$ & $\mu_{10}=0$ & $\mu_{10} = 0$ & $\mu_{01} = 0$ \\
  	 	& & $\beta \mu_{01} \leq \mu_{11}$ & $\mu_{10} = 0$\\
  	 \cline{2-4}
  	\end{tabular}
  	\caption{The set $C_{CPT}$ for games of type (I) with weakly dominated strategies}\label{fig: dom_case1}
  \end{figure} 

\begin{figure}[t]
\centering
  	\begin{tabular}{c | c | c |c |}
  	    \multicolumn{1}{c}{}& \multicolumn{1}{c}{$\beta = 0$} & \multicolumn{1}{c}{$0 < \beta < \infty$} & \multicolumn{1}{c}{$\beta = \infty$}\\
  	    \cline{2-4}
  	 $\alpha = 0$ & $\mu_{11} = 0$ & $\mu_{11} =   \mu_{01} = 0$ & $\mu_{11} = 0$ \\
  	 	& $\mu_{10} = 0$ & $\beta \mu_{00} \geq \mu_{10}$ & $\mu_{01} = 0$\\
  	 \cline{2-4}
  	 $0 < \alpha < \infty$ & $\mu_{11} = \mu_{10} = 0$ & $-$ & $\mu_{01} = \mu_{00} = 0$ \\
  	 	& $\alpha \mu_{00} \leq \mu_{01}$ &  & $\alpha \mu_{10} \geq \mu_{11}$\\ 
  	 \cline{2-4}
  	 $\alpha = \infty$ & $\mu_{10}  =0$ & $\mu_{00} =\mu_{10} = 0$ & $\mu_{00} =  0$ \\
  	 	& $\mu_{00} = 0$ & $\beta \mu_{01} \leq \mu_{11}$ & $\mu_{01} = 0$\\
  	 \cline{2-4}
  	\end{tabular}
  	\caption{The set $C_{CPT}$ for games of type (II) with weakly dominated strategies}\label{fig: dom_case2}
  \end{figure} 

  The geometry of $C_{CPT}$ in case (I) is as follows:
 \begin{itemize}
 	\item If $\alpha = \beta = 0$ or $\alpha = \beta = \infty$, then $C_{CPT}$ is a line with endpoints $F = (1,0,0,0)$ and $G = (0,0,0,1)$. It intersects the set $I$ at the two endpoints $F$ and $G$.
 	\item If $\alpha=0,\beta = \infty$, then $C_{CPT}$ is a triangle with vertices $F = (1,0,0,0)$, $G = (0,0,1,0)$, and $H = (0,0,0,1)$. It intersects $I$ at the lines with endpoints  $\{F,G\}$ and $\{G,H\}$. Similarly, if $\alpha=\infty,\beta = 0$, then $C_{CPT}$ is a triangle and it intersects $I$ at two lines.
 	\item If $\alpha = 0, 0 < \beta < \infty$, then $C_{CPT}$ is a triangle with vertices $F = (0,0,0,1)$, $G = (1,0,0,0)$, and $H = (\frac{1}{1+\beta},0,\frac{\beta}{1+\beta},0)$. It intersects the set $I$ at the point $(0,0,0,1)$ and the line joining the points $(1,0,0,0)$ and $(\frac{1}{1+\beta},0,\frac{\beta}{1+\beta},0)$. The remaining three cases can be analyzed similarly.
 \end{itemize}
The geometry of $C_{CPT}$ in case (II) is as follows:
\begin{itemize}
	\item If $\alpha = \beta = 0$, then the set $C_{CPT}$ is a line joining the points $(1,0,0,0)$ and $(0,1,0,0)$ contained in the set $I$. Similarly, for the cases when $\alpha = 0$ or $\infty$ and $\beta = 0$ or $\infty$, the set $C_{CPT}$ is a line segment contained in the set $I$.
	\item If $\alpha = 0, 0 < \beta < \infty$, then $C_{CPT}$ is a line joining the points $(1,0,0,0)$ and $(\frac{1}{1+\beta},0,\frac{\beta}{1+\beta},0)$ and is contained in the set $I$. The remaining three cases can be analyzed similarly.
\end{itemize}
The geometry of $C_{CPT}$ is cases (IV) and (III) can be obtained from cases (I) and (II) respectively, using the transformation $\tau$, and replacing $\beta$ by $1/\beta$.

We now consider a $2 \times 2$ game with at least one equivalent pair of strategies. Suppose player $1$ has equivalent strategies. This corresponds to case (i) above with both equalities (i.e. $a_{00} = a_{10}$ and $a_{01} = a_{10}$). Thus player $1$ is indifferent between his strategies. For player $2$, if the two strategies are equivalent, then the game is trivial and $C_{CPT} = \Delta^3$. If one of the strategies for player $2$ is weakly dominated, say strategy $0$ is weakly dominated by strategy $1$, then either $b_{00} = b_{01}, b_{10} < b_{11}$ or $b_{00} < b_{01}, b_{10} = b_{11}$. If $b_{00} = b_{01}, b_{10} < b_{11}$, then the set $C_{CPT} = \{\mu \in \Delta^3| \mu_{10} = 0\}$ is a triangle with vertices $F = (1,0,0,0)$, $G = (0,1,0,0)$ and $H = (0,0,0,1)$. It intersects the set $I$ at the lines with endpoints $\{F,G\}$ and $\{G,H\}$. The other three cases are similar. If neither of the two strategies for player $2$ dominates the other, then $C_{CPT}$ is characterized by the inequalities 
\begin{equation*}\label{eq: beta_ineq}
	\{\beta \mu_{00} \geq  \mu_{10},\; \beta \mu_{01} \leq \mu_{11}\} \text{ or } \{\beta \mu_{00} \leq  \mu_{10},\; \beta \mu_{01} \geq \mu_{11}\}.
\end{equation*}
where the former pair holds if $b_{00} > b_{01}, b_{10} < b_{11}$ and the latter holds if $b_{00} < b_{01}, b_{10} > b_{11}$.
Suppose the first pair of inequalities hold (the other case can be handled similarly). Then one can check that the set $C_{CPT}$ is a tetrahedron with vertices $E = (1,0,0,0), F = (\frac{1}{1+\beta},0,\frac{\beta}{1+\beta},0), G = (0,0,0,1)$ and $H = (0,\frac{1}{1+\beta},0,\frac{\beta}{1+\beta})$. It intersects $I$ at the lines with endpoints $\{E,F\}$, $\{G,H\}$ and $\{F,H\}$.


\section[connected]{Connectedness of $C_{CPT}$}\label{sec: conn}

In the previous section, we saw that for $2 \times 2$ games the set $C_{CPT}$ is a convex polytope. However, in general, the set $C_{CPT}$ can have a more complicated geometry. We will now see that the set $C_{CPT}$ can, in fact, be  disconnected.

In this section, we restrict our attention to games with each player $i$ having reference point $r_i = 0$, and all the outcomes $h_i(\cdot)$ non-negative. Thus all our outcome profiles are ``one-sided'' with zero reference point, and we will denote $w_i^+(\cdot), v_i^0(\cdot), V_i^r(\cdot)$ simply by $w_i(\cdot),v_i(\cdot),V_i(\cdot)$ respectively.

The geometry of the set $C_{CPT}$ is determined by the set of inequalities~(\ref{eq: CPT_corr_ineq}). Let us consider the inequality corresponding to player $i$ deviating from strategy $s_i$ to $d_i$. For ease of notation, fix a one to one correspondence between the numbers $\{1,\dots,t\}$ and the joint strategies $\{s_{-i} \in S_{-i}\}$ (here $t = |S_{-i}|$). Let 
\[
	x = (x_1, \dots, x_t) := \l(v_i(h_i(s_i,s_{-i}))\r)_{s_{-i} \in S_{-i}},
\]
 and 
 \[
 	y = (y_1,\dots,y_t) := \l(v_i(h_i(d_i,s_{-i}))\r)_{s_{-i} \in S_{-i}}.
 \]
  Let $p = (p_1,\dots,p_t) \in \Delta^{t-1}$ be a joint probability distribution on $S_{-i}$.
  Let $(a_1,\dots,a_t)$ and $(b_1,\dots,b_t)$ be permutations of $(1,\dots,t)$ such that
\[
	x_{a_1} \geq x_{a_2} \geq \dots \geq x_{a_t} \text{ and } y_{b_1} \geq y_{b_2} \geq \dots \geq y_{b_t},
\]
respectively.

Consider the inequality
\begin{equation}\label{eq: simple_notation_ineq}
	\tilde{V}_i(p,x) \geq \tilde{V}_i(p,y),
\end{equation}
where
\begin{align}\label{eq: conn_cpt_x}
	\tilde{V}_i(p,x) &= x_{a_t} + w_i(p_{a_1} + \dots + p_{a_{t-1}})[x_{a_{t-1}} - x_{a_t}] \nonumber\\
	&+ w_i(p_{a_1} + \dots + p_{a_{t-2}})[x_{a_{t-2}} - x_{a_{t-1}}]
	+ \dots + w_i(p_{a_1})[x_{a_1} - x_{a_2}], 
\end{align}
and
\begin{align}\label{eq: conn_cpt_y}
	\tilde{V}_i(p,y) &= y_{b_t} + w_i(p_{b_1} + \dots + p_{b_{t-1}})[y_{b_{t-1}} - y_{b_t}]\nonumber\\
 	& + w_i(p_{b_1} + \dots + p_{b_{t-2}})[y_{b_{t-2}} - y_{b_{t-1}}] + \dots + w_i(p_{b_1})[y_{b_1} - y_{b_2}].
\end{align}
{To contrast with the notation used in earlier sections, note
that $\tilde{V}_i(p,x) = V_i(p,\tilde{x})$ and $\tilde{V}_i(p,y) = V_i(p,\tilde{y})$, where
$\tilde{x} := \l(h_i(s_i,s_{-i})\r)_{s_{-i} \in S_{-i}}$ and
$\tilde{y} := \l(h_i(d_i,s_{-i})\r)_{s_{-i} \in S_{-i}}$.}
Let $C(i,s_i,d_i)$ denote the set of all probability vectors $p \in \Delta^{t-1}$ that satisfy 
the inequality~\eqref{eq: simple_notation_ineq}.
We can similarly define $C(i,s_i,d_i)$ for all $i \in N, s_i,d_i \in S_i$. It is clear from the definition of CPT correlated equilibrium that for a joint probability distribution $\mu \in C_{CPT}$, provided $\mu_i(s_i) > 0$, the probability vector $p  = \mu_{-i}^{s_i} \in \Delta^{t-1}$ should belong to $C(i,s_i,d_i)$ for all $d_i \in S_i$. Let
\[
	C(i,s_i) := \cap_{d_i \in S_i} C(i,s_i,d_i).
\]
Now, for all $i$, define a subset $C(i) \subset \Delta^{|S|-1}$, as follows:
\[
	C(i) := \{\mu \in \Delta^{|S|-1} | \mu_{-i}^{s_i} \in C(i,s_i), \forall s_i \in S_i \text{ such that } \mu_i(s_i) > 0 \}.
\]
Note that since $C_{CPT}$ is nonempty, the set $C(i)$ is nonempty for each $i$. The set $C(i)$ can be constructed from the sets $\{C(i,s_i),s_i \in S_i\}$ as follows: let $p^{s_i} \in C(i,s_i)$ for all $s_i \in S_i$ such that $C(i,s_i) \neq \phi$, let $q_i \in \Delta^{|S_i|-1}$ be a probability distribution over $S_i$ such that $q_i(s_i) =  0$ for all $s_i  \in S_i$ such that 
$C(i,s_i) = \emptyset$, 
and define a joint probability distribution $\mu \in \Delta^{|S|-1}$ by $\mu(s_i,s_{-i}) = q_i(s_i) p^{s_i}(s_{-i})$ if $C(i,s_i) \neq \phi$ and $\mu(s_i,s_{-i}) = 0$ otherwise. Then $\mu \in C(i)$, and for every $\mu \in C(i)$, the corresponding $q_i = \mu_i$ for all $s_i \in S_i$  and $p^{s_i} = \mu_{-i}^{s_i}$ for all $s_i \in S_i$ with $C(i,s_i) \neq \phi$. Further, it is clear that 
\[
	C_{CPT} = \cap_{i \in N} C(i).
\]
Thus the set $C_{CPT}$ is uniquely determined by the collection of sets 
$$\{C(i,s_i,d_i),i\in N,s_i,d_i \in S_i\}.$$

\begin{lemma}
	In the above setting, the set $C(i,s_i,d_i)$ is connected.
\end{lemma}
\begin{proof}
Suppose the permutations $(a_1,\dots,a_t)$ and $(b_1,\dots,b_t)$ can be chosen such that they are equal. Let
\begin{align}\label{eq: p_to_l}
	l_j := w_i(\sum_{k=1}^j p_{a_k}) = w_i(\sum_{k=1}^j p_{b_k}), \text{ for } 1 \leq j \leq t.
\end{align}
For every vector $l = (l_1,\dots,l_t) \in \bbR^t$ such that $0 \leq l_1 \leq \dots \leq l_t = 1$, there corresponds a unique probability vector $p = (p_1,\dots,p_t)$ satisfying equations~(\ref{eq: p_to_l}) and this mapping is continuous because $w_i(\cdot)$ is a continuous strictly increasing function. Thus we have a one-to-one correspondence between probability vectors $(p_1,\dots,p_t)$ and the vectors $(l_1,\dots,l_t)$.

Inequality~(\ref{eq: simple_notation_ineq}) can then be written as
\begin{equation}\label{eq: ineq_l_linear}
	l_t x_{a_t} + \sum_{i=1}^{t-1} l_{t-i}[x_{a_{t-i}} - x_{a_{t-i+1}}] \geq l_t y_{b_t} + \sum_{i=1}^{t-1} l_{t-i}[y_{b_{t-i}} - y_{b_{t-i+1}}].
\end{equation}
	
Since this is linear in $(l_1,\dots,l_t)$, the set of all vectors $(l_1,\dots,l_t)$ satisfying inequality~(\ref{eq: ineq_l_linear}) is a convex polytope. In particular, it is connected. Thus the set $C(i,s_i,d_i)$ is also connected.

Suppose now the permutations $(a_1,\dots,a_t)$ and $(b_1,\dots,b_t)$ cannot be chosen to be equal. Then there exists $1\leq j_1,j_2 \leq t$ such that $x_{j_1} > x_{j_2}$ and $y_{j_1} \leq y_{j_2}$. If $p \in C(i,s_i,d_i)$ such that $p_{j_2} > 0$, then, by the stochastic dominance property, the following probability vector $q(\epsilon)$, for all $0 \leq \epsilon \leq 1$, also belongs to $C(i,s_i,d_i)$:
\begin{align*}
	q_j(\epsilon) = \begin{cases}
		p_{j_1} + (1-\epsilon) p_{j_2} \text{ if } j = j_1,\\
		\epsilon p_{j_2} \text{ if } j = j_2,\\
		p_j \text{ otherwise}.
	\end{cases}
\end{align*}
Thus, from every vector $p \in C(i,s_i,d_i)$, we have a path connecting it to a probability vector $p' \in C(i,s_i,d_i)$ such that $p'_{j_2} = 0$. To show that $C(i,s_i,d_i)$ is connected it is enough to show that the subset
\[
	C'(i,s_i,d_i) = \{p' \in C(i,s_i,d_i)|p'_{j_2} = 0\}.
\]
is connected. From (\ref{eq: conn_cpt_x}) and (\ref{eq: conn_cpt_y}), we can see that the CPT values of the prospects $(p,x)$ and $(p,y)$ with probability vector restricted to $C'(i,s_i,d_i)$ do not depend on the outcomes $x_{j_2}$ and $y_{j_2}$. If one can now choose permutations $(a'_1,\dots,a'_{t-1})$ and $(b'_2,\dots,b'_{t-1})$ of $\{1,\dots,t\}\back \{j_2\}$ such that
\[
	x_{a'_1} \geq x_{a'_2} \geq \dots \geq x_{a'_{t-1}} \text{ and } y_{b'_1} \geq y_{b'_2} \geq \dots \geq y_{b'_{t-1}},
\]
then, as before, one can argue that the set $C'(i,s_i,d_i)$ is connected. If not, we can continue to decrease the support of the probability vectors under consideration. This process terminates since our state space is finite.
\end{proof}

Even though the sets $C(i,s_i,d_i)$ are connected, their intersection might be disconnected, as in example~\ref{ex: disconn}.


\begin{example}\label{ex: disconn}
 Consider a $2$ player game with each player having three pure strategies: TOP, CENTER, BOTTOM  for player $1$ (row player) and RED, YELLOW, GREEN for player $2$ (column player), with the corresponding payoffs as shown in table~\ref{tab: 3x3game}. For both the players, let $v_i(\cdot)$ be the identity function. For the probability weight function $w_i(\cdot)$ we employ the function suggested by \citet{prelec1998probability}, which, for $i = 1,2$, is given by
\[
	w_i(p) = \exp \{-(-\ln p)^{\alpha_i}\},
\]
for some $\alpha_i \in (0,1]$. We take $\alpha_1 = 0.5$ and $\alpha_2 = 1$.
\begin{table}
\centering
\begin{tabular}{c | c | c | c|}
	 \multicolumn{1}{c}{}	& \multicolumn{1}{c}{RED}   &  \multicolumn{1}{c}{YELLOW}  &  \multicolumn{1}{c}{GREEN}   \\
	 \cline{2-4}
	 TOP	& $69,10$ & $61,0$ & $20,10$\\
	 \cline{2-4}
	 CENTER	& $50,0$ & $60,10$ & $30,0$\\
	 \cline{2-4}
	 BOTTOM	& $101,0$ & $41,10$ & $0,0$\\
	 \cline{2-4}
	 \end{tabular} 
	 \caption{Payoff matrix for the game in example~\ref{ex: disconn}}\label{tab: 3x3game}
\end{table}
We will now see that the set $C(1,\text{TOP})$ is disconnected. Fix the correspondence $(R,Y,G) \leftrightarrow (\text{RED, YELLOW, GREEN})$. The set $C(1,\text{TOP},\text{BOTTOM})$ consists of all probability vectors $p = (p_R,p_Y,p_G) \in \Delta^2$ satisfying the following inequality:
\begin{align*}
	20 &+ w_1(p_R + p_Y)[61 - 20] + w_1(p_R)[69 - 61]\\
	  & \geq 0 + w_1(p_R + p_Y)[41 - 0] + w_1(p_R)[101 - 41].
\end{align*}
This holds iff $p_R \leq 0.40$ (all the decimal numbers henceforth are correct up to two decimal points). Thus, we have
\[
	C(1,\text{TOP,BOTTOM}) = \{p \in \Delta^2 | p_R \leq 0.40\}.
\]
The set $C(1,\text{TOP},\text{CENTER})$ consists of all probability vectors $p = (p_R,p_Y,p_G) \in \Delta^2$ 
satisfying the inequality
\begin{align*}
	20 &+ w_1(p_R + p_Y)[61 - 20] + w_1(p_R)[69 - 61]\\
	& \geq 30 + w_1(p_R + p_Y)[50 - 30] + w_1(p_Y)[60 - 50].
\end{align*}
Rearranging, we get
\[
	21w_1(1 - p_G) - 10w_1(1 - p_R - p_G) \geq 10 - 8w_1(p_R).
\]
For each $p_R \in [0,0.4]$, we solve the above inequality for $p_G$. The set $C(1,\text{TOP})$, as shown in figure~\ref{fig: not_conn}, is disconnected.
\begin{figure}
	\includegraphics[scale = 0.3]{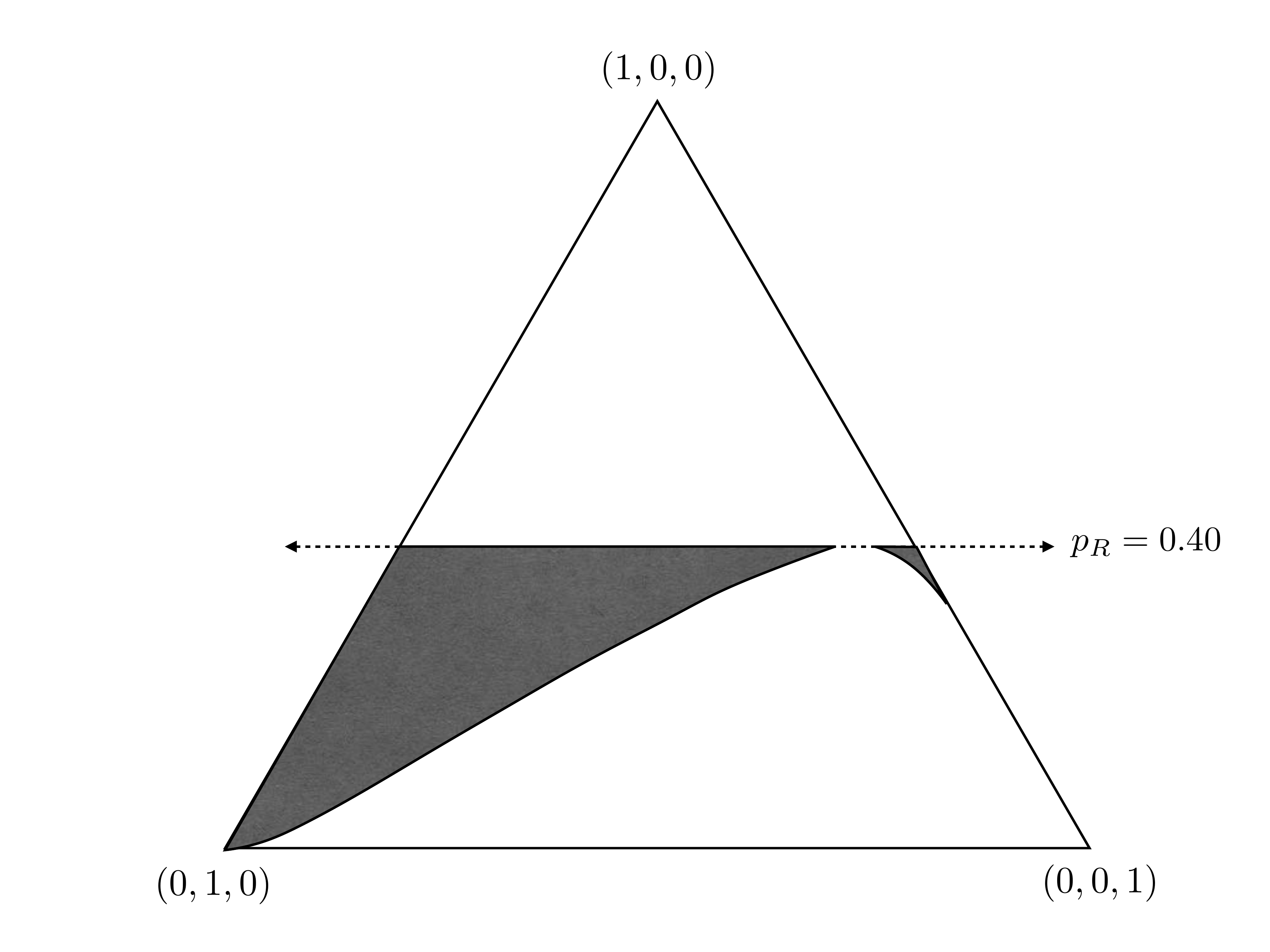}
	\caption{Standard $2$-simplex of probability vectors $p = (p_R,p_Y,p_G)$. The shaded region represents the set $C(1,\text{TOP})$ and is disconnected.}\label{fig: not_conn}
\end{figure}
One can check that 
\[
	(0,\epsilon,1-\epsilon) \in C(1,\text{CENTER}) \text{ and } (1-\epsilon,\epsilon,0) \in C(1,\text{BOTTOM}),
\] 
for $\epsilon \in [0,0.20]$. We cannot as yet conclude that the set $C(1)$ is disconnected, because of the existence of joint probability distributions $\mu$ with marginal distribution $\mu_1(\text{TOP}) = 0$. We now show that $C(2)$ cannot contain any distribution $\mu$ with $\mu_1(\text{TOP}) = 0$.

Fix the correspondence $(T,C,B) \leftrightarrow (\text{TOP,CENTER,BOTTOM})$. A similar analysis for player $2$ shows that
\begin{align*}
	C(2,\text{RED}) &= \{p \in \Delta^2 | p_T \geq 0.5\},\\
	C(2,\text{YELLOW}) &= \{p \in \Delta^2 | p_T \leq 0.5\},\\
	C(2,\text{GREEN}) &= \{p \in \Delta^2 | p_T \geq 0.5\}.
\end{align*}
Suppose that there were $\mu \in C_{CPT}$ with $\mu_{1}(\text{TOP}) = 0$. Then 
	$$\mu(\text{TOP,RED}) = \mu(\text{TOP,YELLOW}) = \mu(\text{TOP,GREEN}) = 0,$$
and from the structure of the sets $C(2,\text{RED})$ and $C(2,\text{GREEN})$ we get 
$$\mu_2(\text{RED}) = \mu_2(\text{GREEN}) = 0.$$ Thus, the joint probability $\mu$ has support only on the strategy pairs (CENTER,YELLOW) and (BOTTOM,YELLOW). Thus, player $2$ always plays strategy YELLOW and player $1$ mixes between CENTER and BOTTOM. However, given player $2$ plays strategy YELLOW, player $1$'s TOP strategy dominates strategies CENTER and BOTTOM. Hence such a joint probability distribution is not possible. Thus there does not exist any distribution $\mu \in C_{CPT}$ with $\mu_{1}(\text{TOP}) = 0$.

There is a possibility that one of the components of $C(1,\text{TOP})$ could disappear in the intersection $C(1) \cap C(2)$. However, this does not happen because both the distributions $\bar \mu, \tilde \mu$ in figure~\ref{fig: sample_dist} belong to $C_{CPT}$ with $\bar \mu_{-1}^{\text{TOP}}$ and  $\tilde \mu_{-1}^{\text{TOP}}$ belonging to different components of $C(1,\text{TOP})$.
\begin{figure}
\centering
\begin{tabular}{c | c | c | c|}
	 \multicolumn{1}{c}{}	& \multicolumn{1}{c}{RED}   &  \multicolumn{1}{c}{YELLOW}  &  \multicolumn{1}{c}{GREEN}   \\
	 \cline{2-4}
	 TOP	& $0.4$ & $0.1$ & $0.5$\\
	 \cline{2-4}
	 CENTER	& $0$ & $0.05$ & $0.5$\\
	 \cline{2-4}
	 BOTTOM	& $0.4$ & $0.05$ & $0$\\
	 \cline{2-4}
	 \end{tabular} 
	 \vskip\baselineskip
	 \centering
	 \begin{tabular}{c | c | c | c|}
	 \multicolumn{1}{c}{}	& \multicolumn{1}{c}{RED}   &  \multicolumn{1}{c}{YELLOW}  &  \multicolumn{1}{c}{GREEN}   \\
	 \cline{2-4}
	 TOP	& $0.4$ & $0$ & $0.6$\\
	 \cline{2-4}
	 CENTER	& $0$ & $0$ & $0.6$\\
	 \cline{2-4}
	 BOTTOM	& $0.4$ & $0$ & $0$\\
	 \cline{2-4}
	 \end{tabular} 
	 \caption{Un-normalized distributions $\bar \mu$ and $\tilde \mu$.}\label{fig: sample_dist}
\end{figure}

\end{example}


\section{Conclusions and future work}

Although the set of correlated equilibria under CPT has a more complicated geometry than a convex polytope, property (P), on the intersection of the Nash and correlated equilibrium sets, continues to hold. Property (P) is particularly relevant to the interactive learning problem in game theory \citep{foster1997calibrated,foster2006regret,hart2000simple}. This raises the question of analyzing the interactive learning problem under cumulative prospect theoretic preferences. We leave this for future work.


\bibliographystyle{abbrvnat} 
\bibliography{Bib_Database}

\end{document}